\documentclass{article}
\usepackage[letterpaper, portrait, margin=1in]{geometry}
\usepackage{graphicx, caption}
\usepackage{amsfonts}
\usepackage{amsmath,bm}
\usepackage{cite}
\usepackage{amssymb}
\usepackage{amsthm}
\usepackage{authblk}
\usepackage{color}
\usepackage{mathtools}
\usepackage{braket}
\usepackage[shortlabels]{enumitem}

\theoremstyle{plain}
\newtheorem{thm}{Theorem}

\newtheorem*{defn}{Definition} 

\newtheorem{proposition}[thm]{Proposition}

\newtheorem{corollary}[thm]{Corollary}

\newtheorem*{remark}{Remark}

\newcommand\blfootnote[1]{%
	\begingroup
	\renewcommand\thefootnote{}\footnote{#1}%
	\addtocounter{footnote}{-1}%
	\endgroup
}

\begin{document}

\title{Increasing the classical data throughput in quantum networks by combining quantum linear network coding with superdense coding}
\author[1,2]{Steven Herbert}
\affil[1]{\small{\textit{Department of Computer Science, University of Oxford, UK}}}
\affil[2]{\small{\textit{Riverlane, 1st Floor St Andrews House,
59 St Andrews Street, Cambridge,  UK}}}

\maketitle

\begin{abstract}
\noindent \blfootnote{Contact: sjh227@cam.ac.uk}This paper shows how network coding and superdense coding can be combined to increase the classical data throughput by a factor $2-\epsilon$ (for arbitrarily small $\epsilon > 0$) compared to the maximum that could be achieved using either network coding or superdense coding alone. Additionally, a general decomposition of a ``mixed'' network (i.e., consisting of classical and quantum links) is given, and it is reasoned that, owing to the inherent hardness of finding network codes, this may well lead to an increase in classical data throughput in practise, should a scenario arise in which quantum networks are used to transfer classical information.
\end{abstract}

\section{Introduction}
\label{intro}
The ability to perform \textit{superdense coding}, where two classical bits can be transmitted using a single qubit when a Bell pair is shared between the transmitter and receiver \cite{superdense}, is one of the fundamental results in the field quantum information. Discoveries such as \textit{quantum teleportation} \cite{teleport} and superdense coding demonstrate something of the fundamental nature of quantum information and have paved the way for modern quantum information theory, in which entanglement is frequently treated as a resource which can aid information transfer. In practise, this resource (shared entanglement) would need to be continually replenished, and so it is reasonable that the equipment and time required to do so should be should be taken into account in any quantification of the benefits of implementing superdense coding in any real world communications system. This reveals two possible operational advantages of superdense coding:
\begin{itemize}
\item \textbf{Superdense coding to increase peak data-rate.}\\ 
If a pair of nodes (say nodes ``A'' and ``B'') are connected by a quantum channel and experience variable classical data-rate demand, then, assuming that at least one of the nodes (say node ``A'') can generate entangled pairs, in the periods of low demand node A can transfer entanglement (namely halves of entangled pairs) to node B which can then be stored (assuming sufficiently high fidelity memory is available at each node) to enable superdense coding at times of high demand, doubling the peak data-rate.
\item \textbf{Superdense coding to ``reverse'' directed edges.}\\  
In the case where the two nodes are connected by a pair of \textit{directed} edges (with equal data-rates), one in each direction then, again assuming the ability of entanglement generation within the nodes, then a data-rate in the desired ``forward'' direction equal to twice the data-rate of the single forward link can be achieved by using the ``backward'' link to continually transfer halves of entangled pairs, as shown in Fig.~\ref{f1}.
\end{itemize}
\begin{figure}[!t]
	\centering
	\includegraphics[width=0.365\linewidth]{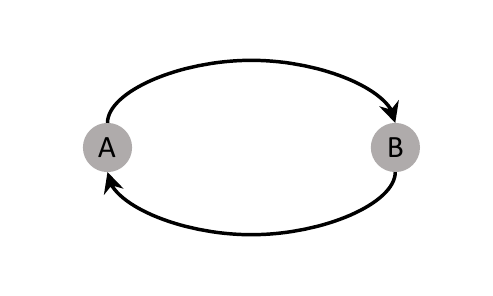}
	\captionsetup{width=0.95\linewidth}
	\caption{\small{Node ``A'' and node ``B'' are connected by two directed links of unit (quantum) data-rate (one in each direction). It is possible to achieve a classical information transfer from A to B at twice the unit data-rate if B continually generates Bell states, and transfers one half thereof to A over the B$\to$A link, as this shared entanglement can be used as a resource to enable superdense coding over the A$\to$B link.}}
	\label{f1}
\end{figure}
\begin{figure}[!t]
	\centering
	\includegraphics[width=0.73\linewidth]{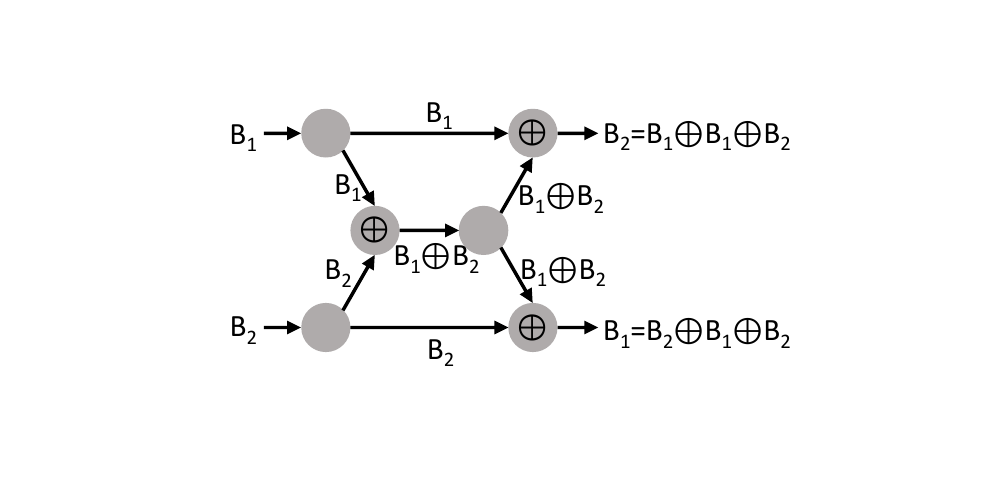}
	\captionsetup{width=0.95\linewidth}
	\caption{\small{Example of network coding over the Butterfly network for input bitstreams ``B$_1$'' and ``B$_2$'' -- nodes either perform a modulo-2 sum of the incoming bitstreams (when labelled $\oplus$) or fanout the single incoming bitstream otherwise -- thus connecting diagonally opposite corners of the directed network simultaneously, which would not be possible by simply routing.}}
	\label{f2}
\end{figure}
The requirement that edges be directed, as in the second item above, prompts consideration of \textit{network coding} and the question of whether it is possible to combine network coding and superdense coding in a way which yields some data throughput increase. Network coding was first introduced by Ahlswede \textit{et al} \cite{Ahlswede2000}, and is a technique in which each node in a telecommunications network combines its incoming bitstreams before forwarding them in such a way that the overall data throughput exceeds that which can be achieved by simply forwarding the bitstreams. The standard illustrative example of network coding is the Butterfly network, shown in Fig.~\ref{f2}. Whilst network coding can in principle be applied to undirected as well as directed networks, the famous \textit{multiple unicast conjecture} holds that in undirected networks the data throughput achieved by network coding can always be achieved by the alternative technique of fractional routing (the details of which are not of concern here) \cite{Li2004b}. For this reason, along with the aforementioned advantage of superdense coding over directed networks, this paper is concerned with directed networks (i.e., networks in which all of the links are unidirectional, but with parallel oppositely directed edges permitted) and also considers only the ``$k$-pairs'' setting, in which $k$ bitstreams must be transmitted from transmitter nodes to receiver nodes (i.e., the Butterfly network shows a 2-pairs network coding problem).\\
\indent The possibility of combining network coding and superdense coding to increase data throughput is motivated by the observation that network coding protocols \textit{can} be applied to quantum information transfer \cite{Leung2006} and entanglement distribution \cite{Kobayashi2009,Kobayashi2011,Satoh2012,deBeaudrap2014,beaudrap2019quantum}. The main result presented in this paper is that combining quantum linear network coding (QLNC) and superdense coding can yield a factor $(2- \epsilon)$ increase in data throughput compared to that which could be achieved using QLNC or superdense coding alone (for arbitrarily small $\epsilon > 0$).
\section{Network architecture}
\label{architecture}
A communications network consists of $n$ nodes, connected by noiseless unidirectional links of specified data-rate. There are two types of links, classical links which transfer classical information, and quantum links which transfer quantum information.
\begin{defn}
The data-rate of a classical link is the maximum \textit{bit transfer rate}.
\end{defn}
\begin{defn}
The data-rate of a quantum link is the maximum \textit{qubit transfer rate}.
\end{defn}
\noindent In particular a qubit could be one-half of a Bell pair (the other half of which remains at the transmitter node). Classical protocols over networks including quantum links could be implemented as a special case by transferring only (unentangled) orthogonal basis states, however classical information cannot be communicated over the quantum links faster than the data-rate and, unlike the typical assumption in quantum computing networks, there is no mechanism for global classical control -- any classical control information must be communicated through the network.\\
\indent The network consists of $k$ transmitter nodes, each of which must transfer a (classical) bitstream to its corresponding receiver (i.e., there are $k$ receivers in total). All transmitters and receivers are distinct nodes, and additionally there are ``relay'' nodes that are neither transmitters or receivers. Whilst the techniques developed by de Beaudrap and Herbert \cite{beaudrap2019quantum} are used to prove the main result below, in this paper a slightly different setting is considered: the nodes are not single qubits, but rather are small devices hosting a number of qubits and equipped with the necessary (quantum) computational power and resources to enable the desired operations in communications network. The goal is to maximise the data throughput, which for simplicity is defined:
\begin{defn}
The data throughput is the average bit-rate achieved over the $k$ transmitter-receiver pairs.
\end{defn}
%
%
%
\section{Main results}
\label{main}
The essential idea is to construct a composite network in which a ``backward'' quantum network continually replenishes entanglement consumed by a ``forward'' quantum network which performs superdense coding: in this way, the network in Fig.~\ref{f1} is an extremely simple example of the kind of network being considered. In particular, in order to demonstrate a significant advantage gained by combining network coding and superdense coding, a network based on \cite[Fig.~3]{beaudrap2019quantum} is used for the backward network.
\begin{proposition}
\label{prop1}
For any $\epsilon > 0$ there exists a network in which the data throughput can be increased by a factor $(2- \epsilon)$ by using both QLNC and superdense coding, compared to that which can be achieved using either QLNC or superdense coding alone (for some $\epsilon > 0$).
\end{proposition}
\begin{proof}
The proposition is proven by the following construction. Consider a network consisting of $k$ transmitter-receiver pairs, and two intermediate relay nodes. Let the transmitter nodes be denoted $t_1 \cdots t_k$, the receiver nodes $r_1 \cdots r_k$ and the two ``midway'' relay nodes $m_1$ and $m_2$. All links in the network are of unit data-rate, and the network is expressed in terms of two directed graphs over the same set of nodes (vertices), $G_c = \{V,E_c\}$ and $G_q = \{V,E_q\}$, with edges corresponding to the classical and  quantum links in the network respectively. The network connections can be decomposed into seven components:
\begin{enumerate}[(a)]
    \item Each transmitter is connected to its receiver by a quantum link. That is, $\forall i , \{t_i, r_i\} \in E_q$.
    \item Each receiver is connected to each transmitter except its paired transmitter by a quantum link. That is, $\forall i \forall j \neq i, \{r_i, t_j \} \in E_q$.
    \item Each receiver is connected to the second midway node, by a quantum link. That is, $\forall i, \{r_i, m_2\} \in E_q$.
    \item The second midway node is connected to the first midway node by a quantum link. That is, $\{m_2, m_1\} \in E_q$.
    \item The first midway node is connected to each transmitter node by quantum links. That is, $\forall, i \{m_1, t_i\} \in E_q$.
    \item The second midway node is connected to each receiver by a classical link. That is, $\forall i, \{m_2, r_i\} \in E_c$.
    \item Every transmitter is connected to every other transmitter by a classical link. That is, $\forall i \forall j \neq i \{t_i, t_j\} \in E_c$.
\end{enumerate}
Where an ordered pair denotes an edge \textit{from} the vertex corresponding to the first element of the pair \textit{to} the node corresponding to the vertex corresponding to the second element of the pair, and  $1 \leq i,j \leq k$.
\begin{figure}[!t]
	\centering
	\includegraphics[width=0.73\linewidth]{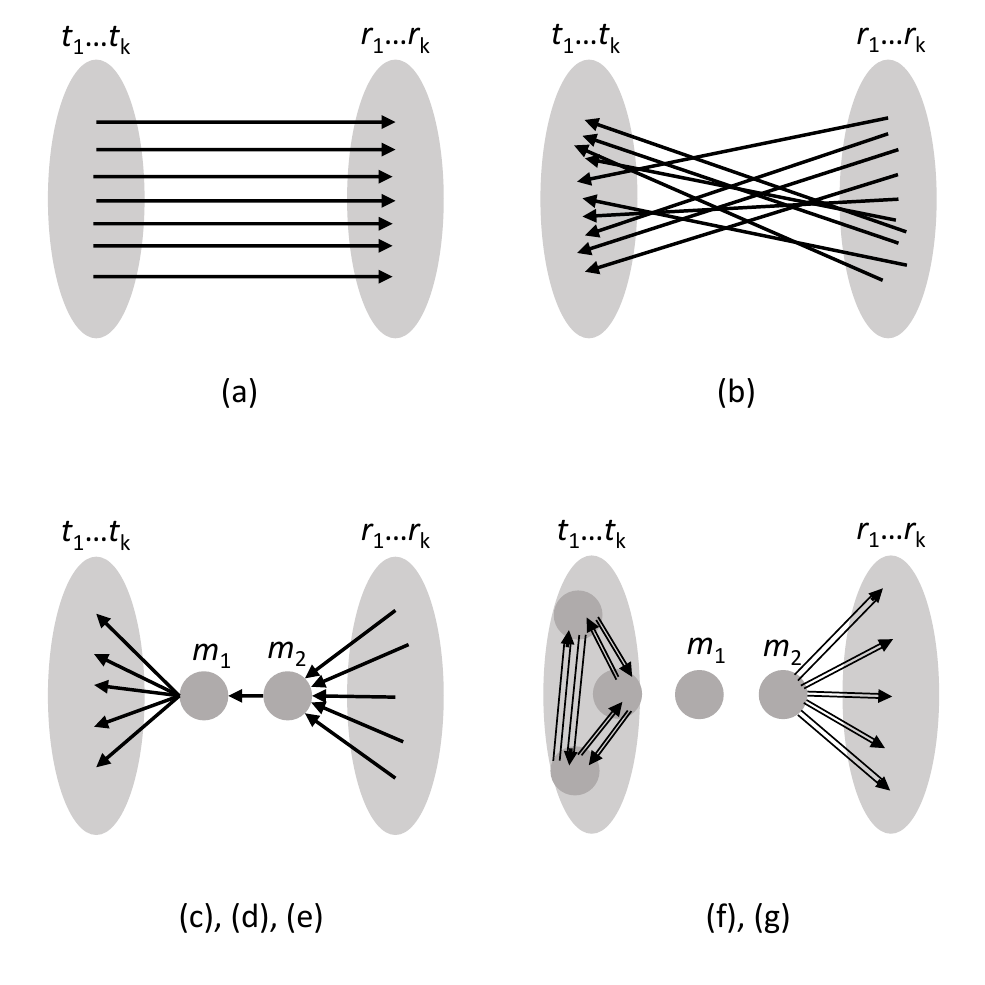}
	\captionsetup{width=0.95\linewidth}
	\caption{\small{The network used to prove Proposition~\ref{prop1}: (a) shows each transmitter connected to its corresponding receiver; (b) shows each receiver connected to every transmitter other than its ``own''; components (c), (d) and (e) show connections from the receivers to the transmitters via a ``bottleneck'' between the intermediate nodes; and components (f) and (g) show classical channels, connecting all pairs of transmitters (in each direction) and from the intermediate node $m_2$ to each receiver.}}
	\label{f3}
\end{figure}
Fig.~\ref{f3} shows a general illustration of such a network. The first part of the proof is showing that components (b) -- (g) admit a quantum linear network code that can distribute a Bell pair between each transmitter-receiver pair, for which the QLNC formalism \cite{beaudrap2019quantum} is used. To recap the QLNC formalism, to the extent which it is used here:
\begin{itemize}
\item Each qubit is labelled by a qubit formula.
\item A qubit can be initialised in either the $\ket{+}$ state, in which case it is given a unique symbolic label, or the $\ket{0}$ state, in which case it is initially labelled 0.
\item The action of a CNOT gate in which a qubit with qubit formula $q_1$ controls a target qubit with qubit formula $q_2$ is to set the qubit formula of the second qubit to $q_1 + q_2$, that is $q_2 \leftarrow q_1 + q_2$ (where all additions are modulo-2, so $q_1 + q_2 \equiv 0$).
\item After the CNOT gates have been executed, any qubit formula which remains as a sum (i.e., $q_1+q_2+q_3$), rather than a single symbol (i.e., $q_1$) is ``terminated'': this is achieved by performing a $X$-basis measurement thereon, and using the result of this measurement to classically control Pauli-$Z$ gates on other qubits such that the qubit formulas of these other qubits sum to the formula of the qubit being terminated. For example, if a qubit with formula $q_1+q_2$ is to be terminated, then the result of the $X$-basis measurement could control a Pauli-$Z$ gate on two qubits, one with formula $q_1+q_3$ and one with formula $q_2+q_3$ (i.e., because $q_1+q_3+q_2+q_3 = q_1+q_2$. Qubits whose formula is a single symbol can be terminated in the same manner, for example if a qubit with formula $q_1$ is to be terminated, then after performing an $X$-basis measurement thereon, the result could be used to classically control a Pauli-$Z$ gate on another qubit labelled $q_1$. Terminated qubits can now be omitted from any further consideration.
\item Once all of the terminations have been performed, what will remain is a number of qubits whose qubit formula is a single symbol. Qubits with the same symbol will be in a $\ket{\Phi^+}$ or $\ket{GHZ}$ state. For example, if qubits 1 and 2 have label $a_1$ and qubits 3 and 4 have label $a_2$, then the state will be $\ket{\Phi^+}\otimes\ket{\Phi^+}$ (with the qubits ordered 1234).
\end{itemize}
With the QLNC formalism thus summarised, consider the following procedure:
\begin{enumerate}[\itshape i.]
    \item Each receiver creates a $(k+1)$-qubit GHZ state (that is, by initialising a single qubit in the $\ket{+}$ state, and using this to control CNOT gates with $k$ other qubits initialised in the $\ket{0}$ as the targets). For the $i$th receiver, let each qubit in the GHZ state be labelled $a_i$. Each receiver sends one of its qubits on each of its outgoing links (i.e., components (b) and (c)). So each receiver $r_i$ has a single qubit labelled $a_i$, node $m_2$ has qubits $a_1 \cdots a_k$, and each transmitter has $k-1$ qubits, such that the $j$th transmitter has qubits labelled $a_1 \cdots a_k$, except $a_j$.
    \item Node $m_2$ initialises a qubit in the state 0, which is then the target of $k$ successive CNOT gates using each of the qubits labelled $a_1 \cdots a_k$ as the control. Thus this qubit is labelled $\sum_{i=1}^k a_i$ after the CNOT gates have been executed (and the $k$ other qubits labelled $a_1 \cdots a_k$ also remain at $m_2$).
    \item At node $m_2$ each of the qubits labelled $a_1 \cdots a_k$ is measured in the $X$-basis. For all $i$, the result of the $X$-basis measurement of the qubit labelled $a_i$ is sent to the node $t_i$ (i.e., over the links in component (f)) whereupon its binary value controls a $Z$ gate on the single qubit located there (i.e., at $t_i$ a qubit labelled $a_i$ resides). Therefore the qubits labelled $a_1 \cdots a_k$ at $m_2$ have been terminated. Node $m_2$ sends its remaining qubit, labelled $\sum_{i=1}^k a_i$ to node $m_1$ (i.e., over the link in component (d)).
    \item Node $m_1$ initialises $(k-1)$ qubits in the state 0, and performs a CNOT controlled by the qubit labelled $\sum_{i=1}^k a_i$ on each as a target. Thus it has a total of $k$ qubits each labelled $\sum_{i=1}^k a_i$. Node $m_1$ sends one of these qubits over each of its outgoing links (i.e., the links in component (e)).
    \item Now each transmitter has $k$ qubits. Specifically, (for all $j$) $t_j$ has a qubit labelled $\sum_{i=1}^k a_i$ (from step \textit{iv}), and $k-1$ qubits with labels $a_1 \cdots a_k$ except for $a_j$ (from step \textit{i}). At all transmitters, the qubit whose formula is a single symbol is used as the control in a CNOT gate performed with the qubit labelled $\sum_{i=1}^k a_i$ as target. Therefore, in the QLNC formalism, at node $t_j$ the qubit formula $\sum_{i=1}^k a_i$ will have added to it $a_i$ for all $i$ except $i=j$, and so the result will be that it will be labelled $a_j$ (i.e., $\sum_{i} a_i + \sum_{i\neq j} a_i = a_j$. The $k-1$ qubits labelled $a_i$ for all $i \neq j$ will be unchanged, and so after these operations each node $t_i$ has $k$ qubits each with a single-symbol label $a_1 \cdots a_k$.
    \item For all $i$, at node $t_i$, a $X$-basis measurement is performed on each qubit except that labelled $a_i$. For all $j$ (except $j=i$) send the of the $X$-basis measurement is sent to node $t_j$ (i.e., via the links in component (g)), whereupon it is used to control a $Z$-gate on the qubit labelled $a_j$, thus terminating all qubits at node $i$, except that with formula $a_i$.
\end{enumerate}
Thus the final state is that, for all $i$, $t_i$ and $r_i$ both host a single qubit labelled $a_i$ and there are no other unterminated qubits. Thus each pair $t_i,r_i$ shares a Bell pair.\\
\indent In the above procedure, steps \textit{ii} and \textit{v} concern only intra-node operations, and are therefore modelled as instantaneous. Thus the procedure takes three units of time to complete. After the initial network code is performed, as each link can be operated independently, and in parallel with the other links, it follows that the QLNC described can be used to distribute a Bell pair between each transmitter-receiver pair at a unit rate. Using the component of the network (a) for the transmission of qubits, each of which encodes two bits using standard superdense coding, from each transmitter to its receiver at a unit rate, and the quantum network code detailed in steps \textit{i} to \textit{iv} to continually replenish entanglement consumed (also at a unit rate, as above), bitstreams of length $n_b$ can be sent from each transmitter $t_i$ to its corresponding receiver $r_i$ in $(n_b/2)+3$ units of time (where the $+3$ accounts for the three time steps initially required to set-up the first network code).\\
\indent The final part of the proof is to show the maximum data throughput that can be achieved by either QLNC or superdense coding alone, relative to that achieved by the combination of network coding and superdense coding detailed above. In the former case, the absense of superdense coding means that there is no mechanism for ``reversing the direction'' of links -- therefore cutting the network such that all of the transmitter nodes are in a first partition, and all other nodes are in a second partition, it can be seen that there are $k$ outgoing links from the first partition to the second, each of unit rate, therefore the average data-rate  for the $k$ transmitter-receiver pairs cannot exceed one unit: i.e., it will take $n_b$ units of time to transmit the bitstream.\\
\indent In the case where superdense coding is allowed, but network coding is not, the only quantum path from any receiver to its transmitter is via the quantum link $\{m_2, m_1\}$, and therefore the rate of entanglement replenishment per transmitter-receiver pair is limited to $\frac{1}{k}$, if shared equally. Therefore, it will only be possible to use superdense coding on $\frac{1}{k}$ of the transmitted qubits, and for the remaining $\frac{k-1}{k}$ only a rate of one bit per qubit will be achievable. So the average rate will be $\frac{k+1}{k}$. Again, setting up the initial shared entanglement will itself take $3$ units of time, and so transmitting the $n_b$ bit bitstreams between each transmitter-receiver pair will take $((k+1)/k)n_b + 3$ units of time. Thus, by choosing sufficiently large $k$ and $n_b$, the claim that data throughput can be increased by a factor $(2- \epsilon)$ by using both QLNC and superdense coding, compared to that which can be achieved using either QLNC or superdense coding alone (for some $\epsilon > 0$) is proven.
\end{proof}
\noindent The network used in the proof of Proposition~\ref{prop1} has been constructed to prove the result, rather than because it resembles anything that may occur in reality: it is a contrived structure, which has a specific bottleneck that can be contravened only by QLNC. It is therefore appropriate to ask whether in practise there may be any benefits to combining superdense coding and QLNC as described in Proposition~\ref{prop1}. One way to answer such a question is not to produce an explicit example of an actual real communications network of a form which will definitively benefit from the result in Proposition~\ref{prop1} but rather to describe how the result enables a more general class of solutions to network connectivity problem, which may prove beneficial in practise. To do so, it is necessary to allow any link in the network to be split into multiple parallel links connecting the same pair of nodes as the original link (in the same direction), with data-rates summing to that of the original link. 
\begin{defn}
A directed network with $n$ nodes can be represented by two weighted multi-graphs with $n$ vertices: the first multi-graph such that the weights of parallel (directed) edges sum to the data-rate of the corresponding quantum link in the network; and the second multi-graph such that the weights of parallel (directed) edges sum to the data-rate of the corresponding classical link in the network. See Fig.~\ref{f3} for a very simple example.
\end{defn}
\begin{figure}[!t]
	\centering
	\includegraphics[width=0.73\linewidth]{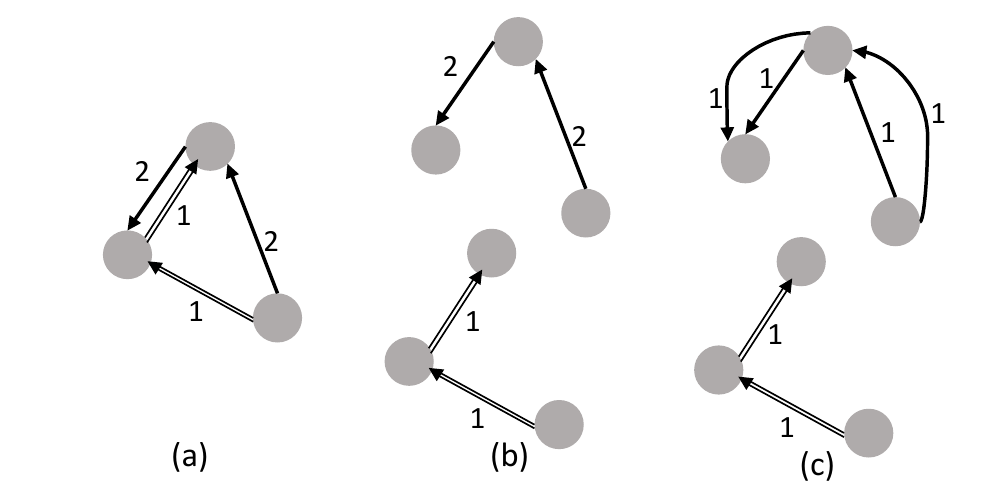}
	\captionsetup{width=0.95\linewidth}
	\caption{\small{Example of drawing equivalent graphs with the link rates shown: (a) shows the network (arrows with double parallel lines are used to denote classical edges); (b) shows the ``natural'' equivalent classical and quantum graphs, where each edge weight is simply the link rate; and (c) shows another equivalent quantum graph with parallel edges included such that each edge has unit weight (in (c) the classical graph is the same as in (b), as all edges have unit weight).}}
	\label{f4}
\end{figure}
\noindent From a communications point of view, this is uncontroversial, as it simply corresponds to sharing the capacity of some link (and for simplicity, the \textit{weight} of an edge will be referred to as its \textit{rate}), and it enables the following corollary to Proposition~\ref{prop1}:
\begin{corollary}
\label{cor2}
Reusing the above notation, let a directed network including $k$ transmitter-receiver pairs (as well as other relay nodes), be represented by a quantum multgraph $G_q$ and a classical multigraph $G_c$ as detailed above, and thereafter be further decomposed into four edge-disjoint components:
\begin{enumerate}
    \item A classical linear network code connecting the $k$-transmitters to their corresponding $k$-receivers, over a subgraph of $G_q \cup G_c$ such that each edge therein has rate $\tilde{w}$.
    \item A subgraph of $G_q$ consisting of edge-disjoint paths from each transmitter to each receiver, in which each edge has rate $w$.
    \item A subgraph of $G_q \cup G_c$, such that each quantum edge consists of a shared Bell pair (unless used to transfer classical information in the QLNC), and each node hosts of an arbitrary number additional qubits, with ``free'' intra-node operations, which admits a classical-quantum linear network code (in the sense of \cite[3.3.3]{beaudrap2019quantum}) connecting $k$ transmitter-receiver pairs, but allowing the binary network code to be such that for any pair the direction of information flow could be either from transmitter to receiver or from receiver to transmitter.
    \item A subgraph of $G_q \cup G_c$ sufficient to perform the terminations of the QLNC in item 3.
\end{enumerate}
Then an asymptotic (with bitstring length) data throughput $\tilde{w} + 2w$ can be achieved.
\end{corollary}
\begin{proof}
The first component straightforwardly allows a rate of $\tilde{w}$ to be achieved, the second third and fourth components enable superdense coding and network coding to be combined achieving a rate of $2w$: the third and fourth components together compose the QLNC which replenishes the entanglement needed to use the second component for superdense coding.
\end{proof}
\begin{remark}
The network decomposition and corresponding data transfer protocol described in Corollary~\ref{cor2} clearly generalises simple routing, network coding (without superdense coding), and superdense coding (without network coding): in the case of simple routing and network coding (without superdense coding), only the first component is used (and for simple routing this will consist of edge-disjoint paths for each transmitter-receiver pair); and in the case of superdense coding (without network coding), the first, second and third components simply consist of edge disjoint paths.
\end{remark}

\section{Discussion}
\label{disc}
The main result of this paper is somewhat removed from the setting with which one may typically be concerned when designing quantum communications networks. In particular, the assertion that data-rate is limited by the information transfer rates of the network's component links and that therefore classical information cannot be transferred faster than quantum information is likely to be unrealistic in the near to medium term -- in reality, the rate of preparation quantum states, especially entangled states is likely to be the ``bottleneck'' in the process. For example, optical fibres can transmit data at rates exceeding 100 terabits per second \cite{fibre}, which is many orders of magnitude above state-of-the-art entanglement generation (e.g., \cite{entanglerate, rate2}). Nevertheless, superdense coding continues to be the subject of theoretical and experimental research \cite{sd1,sd2,sd3}, and given that entanglement generation is a local process, whereas a network of optic fibres is an expensive piece of (spatially-spread) infrastructure, it is envisagable that superdense coding \textit{may} one day be used to increase classical data-throughput.\\
\indent Proceeding with the discussion in this spirit, it can be seen that Corollary~\ref{cor2} is relevant when considering utilising ``mixed'' networks (i.e., consisting of classical and quantum links) for the transfer of classical information, not only because of the potential for a factor of two increase in data throughput, but also because by considering the network as decomposed into the four components of Corollary~\ref{cor2}, more solutions are admitted than by simply considering simple routing, or even network coding. This is important, because finding network codes is in general a NP-hard problem \cite{Lehman2004}, and so it follows that using heuristics to find good solutions is the best that can be hoped for (as opposed to exactly solving). Therefore by explicitly giving the decomposition in Corollary~\ref{cor2}, it \textit{may} be the case that some heuristic can more readily find efficient solutions -- especially given the theoretical potential for higher rates.\\
\indent In summary, this article illustrates how QLNC and superdense coding can theoretically be combined to increase classical data throughput in quantum networks. Moreover, should a scenario ever arise in which superdense coding over quantum networks \textit{is} used to transfer classical information, then for the reasons remarked on below Corollary~\ref{cor2}, the work presented herein is likely to be relevant in practise.

\section*{Acknowledgements}
This work was supported by an Networked Quantum Information Technologies Hub Industrial Partnership Project Grant. The author also thanks Niel de Beaudrap with whom he developed the QLNC formalism that led to this work.

\bibliography{mybib}{}
\bibliographystyle{IEEEtran}

\end{document}